\documentclass{IEEEtran}
\IEEEoverridecommandlockouts
\def\BibTeX{{\rm B\kern-.05em{\sc i\kern-.025em b}\kern-.08em
    T\kern-.1667em\lower.7ex\hbox{E}\kern-.125emX}}
\usepackage{amsmath,amsfonts}
\usepackage{amssymb}
\usepackage{algorithmic}
\usepackage{algorithm}
\usepackage{array}
\usepackage[caption=false,font=footnotesize,labelfont=rm,textfont=rm]{subfig}
\usepackage{textcomp}
\usepackage{stfloats}
\usepackage[caption=false,font=footnotesize]{subfig}
\usepackage{url}
\usepackage{verbatim}
\usepackage{graphicx}
\usepackage{cite}
\usepackage{makecell}
\usepackage{booktabs}
\usepackage{multirow}
\usepackage{color}
\usepackage{hyperref}

\newtheorem{theorem}{Theorem} 
 
\newtheorem{lemma}{Lemma} 
\newtheorem{corollary}{Corollary}

\newtheorem{remark}{Remark}

\usepackage[compact]{titlesec}

\newcommand{\ti}{\text{i}}
\newcommand{\tr}{\text{r}}
\newcommand{\ts}{\text{s}}
\newcommand{\tet}{\text{t}}
\newcommand{\tH}{\text{H}}
\newcommand{\tV}{\text{V}}
\newcommand{\tR}{\text{R}}
\newcommand{\tT}{\text{T}}
\newcommand{\tS}{\text{S}}
\newcommand{\tI}{\text{I}}
\newcommand{\sinc}{\text{sinc}}

\newcommand{\cP}{\mathcal{P}}

\newcommand{\cH}{\mathcal{H}}
\newcommand{\cM}{\mathcal{M}}
\newcommand{\ba}{\mathbf{a}}
\newcommand{\bx}{\mathbf{x}}
\newcommand{\by}{\mathbf{y}}
\newcommand{\bz}{\mathbf{z}}

\newcommand{\bl}{\mathbf{l}}

\newcommand{\bn}{\mathbf{n}}
\newcommand{\bq}{\mathbf{q}}
\newcommand{\bm}{\mathbf{m}}
\newcommand{\bR}{\mathbf{R}}
\newcommand{\bE}{\mathbf{E}}
\newcommand{\bH}{\mathbf{H}}
\newcommand{\bJ}{\mathbf{J}}

\newcommand{\veps}{\varepsilon}
\newcommand{\vtheta}{\vartheta}
\newcommand{\para}{\parallel}
\newcommand{\Ups}{\Upsilon}
\newcommand{\Rayl}{\text{Rayl}}
\newcommand{\rms}{\text{RMS}}
\newcommand{\ITU}{\text{ITU}}

\begin{document}
\title{
A Material-Aware Channel Model for Efficient CKM Generation via Environment Reconstruction
}

\author{
	\IEEEauthorblockN{
		Xuancheng Zhu\IEEEauthorrefmark{1}\IEEEauthorrefmark{2},
        Zhenjun Dong\IEEEauthorrefmark{2},
		Yong Zeng\IEEEauthorrefmark{1}\IEEEauthorrefmark{2},
        Cheng-Xiang Wang\IEEEauthorrefmark{1}\IEEEauthorrefmark{2}\\
	\IEEEauthorblockA{\IEEEauthorrefmark{1}National Mobile Communications Research Laboratory, Southeast University, Nanjing 210096, China\\}
	\IEEEauthorblockA{\IEEEauthorrefmark{2}Purple Mountain Laboratories, Nanjing 211111, China\\}
	Email: 230268975@seu.edu.cn, dongzhenjun@pmlabs.com.cn, \{yong\_zeng,chxwang\}@seu.edu.cn}
}

\maketitle

\begin{abstract}
Channel knowledge map (CKM) is a promising technology for environment-aware wireless communication, sensing, and localization in 6G networks. 
Accurate CKM generation requires precise reconstruction of the environment, including 3D geometries and scatterer materials, typically from multi-modal sensory observations such as LiDAR point clouds and sparse channel measurements. 
While the former is relatively easy to acquire, materials remain difficult to obtain directly from sparse channel measurements due to the lack of an explicit channel model linking them. 
To fill this gap, this paper proposes a material-aware channel model that explicitly characterizes the influence of scatterer materials on the wireless channel. 
Based on this model, an iterative gradient descent based material reconstruction algorithm is proposed. 
Full wave simulation results validate the developed model and the proposed algorithm, demonstrating their potentials for efficient CKM generation via environment reconstruction.
\end{abstract}

\section{Introduction}\label{sec:intro}
Sixth-generation (6G) networks target Tbps level data rates, sub-ms latency, $10^8$ devices/km$^2$ connectivity density, and centimeter-level positioning accuracy \cite{ref:6G}. 
To this end, higher frequency bands, XL-MIMO \cite{ref:XL_MIMO_commun}, and integrated sensing and communication (ISAC) \cite{ref:ISAC} have been developed. 
However, the resulting complexity in XL-MIMO channel estimation, ISAC sensing, and extra high data rate of mmWave/THz systems limits their deployment \cite{ref:channel_estim}.
Recently, environment-aware systems based on channel knowledge maps (CKMs) \cite{ref:CKM} or channel digital twins \cite{ref:radio_pool} have emerged as promising 6G paradigms by exploiting the intrinsic environment and channel relationship. 
Overall, CKM generation methods fall into two categories, including direct end-to-end generation from multi-modal observations via mathematical interpolation \cite{ref:CKM_interpolation}, matrix completion \cite{ref:CKM_matrix_completion}, or deep learning, which requires abundant measurements, and environment reconstruction based generation, which first reconstructs 3D geometry and materials, then applies diffusion models \cite{ref:CKM_diff}, flow matching \cite{ref:CKM_flow_matching}, or ray tracing \cite{ref:sionna} to generate channel knowledge. 
This approach offers \emph{explainability, generalizability, and improved data efficiency}, but material acquisition remains challenging despite the already availability of 3D geometry LiDAR point clouds.

In fact, inferring materials from sparse channel measurements for CKM generation is more promising than recognizing from visuals \cite{ref:visual_mat_reco} or ISAC \cite{ref:ISAC_mat_reco}, as it captures how material impacts on the channel. 
However, this method requires a model explicitly linking material properties to the channel response. 
While statistical models \cite{ref:3GPP} omit materials; the unified channel model \cite{ref:channel_model}, geometry-based stochastic model (GBSM), and pervasive models \cite{ref:perva_geo} incorporate materials only implicitly via radar cross section (RCS), which depends on material, directions, and frequency; deterministic models \cite{ref:det_channel} rely on numerical computation, this paper develops a material-aware channel model based on electromagnetic (EM) theory that \emph{explicitly characterizes scatterer materials via their relative permittivities and conductivities}, and proposes an iterative gradient descent based material reconstruction algorithm. 
Finally, numerical results based on full wave simulation are provided to validate the effectiveness of the proposed models, demonstrating their potentials for effective CKM generation.

\section{System Model}\label{sec:sys_model}
As illustrated in Fig. \ref{fig:sys_model}, we consider a target environment with $S$ scatterers of uniform materials in a 3D Cartesian coordinate, where material of scatterer $s$ is denoted by $\{\veps_s,\sigma_s\}$. 
The transmitter and receiver locate at $\bq_\tT,\bq_\tR\in\mathbb{R}^{3\times1}$, and they are equipped with a single isotropic antenna. 
After clutter cancellation \cite{ref:clutter_cancel}, channel impulse response (CIR) is given by:
\begin{equation}
    h(\bq_\tT,\bq_\tR,\tau)=\sum_{s=1}^Sh_s(\bq_\tT,\bq_\tR)\delta(\tau-\tau_s),
    \label{equ:overall_channel}
\end{equation}
where $h_s(\bq_\tT,\bq_\tR)\in\mathbb{C}$ denotes the component governed by the scatterer $s\in\{1,2,\cdots,S\}$; $\tau_s$ denotes corresponding delay.
The material reconstruction problem is formulated by inferring materials of $S$ uniform scatterers denoted by $\bm\in\mathbb{R}^{2S\times1}$:
\begin{equation}
    \bm=\left[\veps_1,\sigma_1,\veps_2,\sigma_2,\cdots,\veps_S,\sigma_S\right]^T,
    \label{equ:material_set}
\end{equation}
from $N$ spatially sparse channel measurements denoted by:
\begin{equation}
    \cH=\left\{h(\bq_\tT,\bq_1,\tau),h(\bq_\tT,\bq_2,\tau),\cdots,h(\bq_\tT,\bq_N,\tau)\right\},
    \label{equ:sparse_channel}
\end{equation}
where $\bq_n\in\mathbb{R}^{3\times1},n\in\{1,2,\cdots,N\}$ denotes channel measuring positions.
\begin{remark}\label{rem:paper_work}
    \emph{
    Expressing $h_s(\bq_\tT,\bq_\tR)$ in \eqref{equ:overall_channel} as an explicit function of $\{\veps_s,\sigma_s\}$ establishes a direct link between materials and the channel, which enables the reconstruction of $\bm$ from $\cH$ with precise geometries of $S$ scatterers as priori knowledge. 
    }
\end{remark}

\begin{figure}
    \centering
    \includegraphics[width=0.6\linewidth]{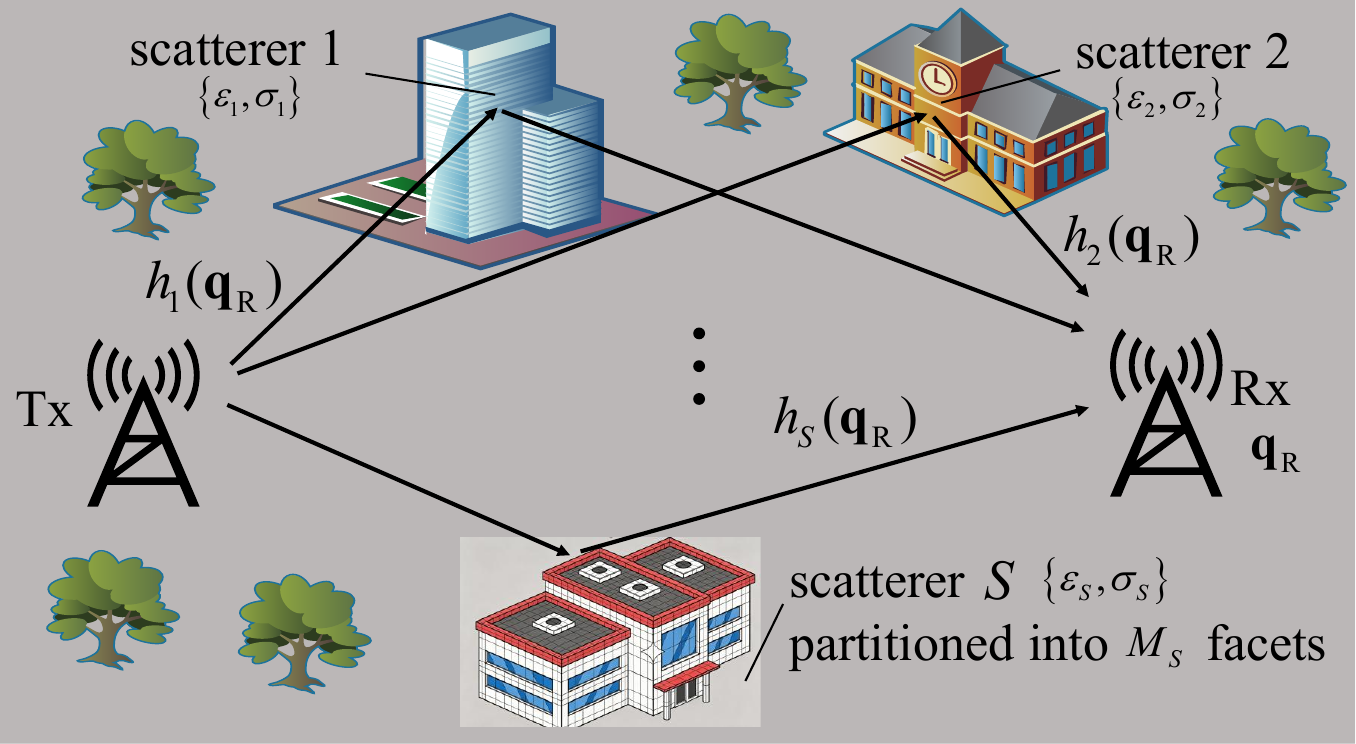}
    \caption{
    A multi-path wireless channel with $S$ scatterers, and the $s$ scatterer is partitioned into $M_s$ rectangular facets.}
    \label{fig:sys_model}
\end{figure}

\section{Material-aware channel model}\label{sec:facet_reflection}
\subsection{Channel response relative to materials}
The channel component $h_s(\bq_\tT,\bq_\tR)$ in \eqref{equ:overall_channel} governed by the scatterer $s$ is expressed by the far-field radar equation \cite{ref:EM_book}:
\begin{equation}
    h_s(\bq_\tT,\bq_\tR)=\sqrt{\frac{P_{\tR,s}}{P_{\text{T}}}}e^{j\varphi_s}=\frac{\lambda\sqrt{\Xi_s(\bq_\tT,\bq_\tR)}e^{j\varphi_s}}{\sqrt{4\pi}(4\pi||\Delta\breve{\bq}_s||\cdot||\Delta\tilde{\bq}_s||)},
    \label{equ:channel_s}
\end{equation}
where $\varphi_s\in(-\pi,\pi]$ denotes the complex phase of $h_s(\bq_\tT,\bq_\tR)$; $P_{\tR,s}/P_\tT$ denotes the ratio of received and emitted power; $\lambda=\frac{c}{f_c}$ is the wavelength and $c$ denotes the light speed; $\Delta\breve{\bq}_s \in \mathbb{R}^{3\times 1}$ points from transmitter towards scatterer $s$, and $\Delta\tilde{\bq}_s = \bq_\tR - \Delta\breve{\bq}_s - \bq_\tT \in \mathbb{R}^{3\times 1}$; $\Xi_s(\bq_\tT,\bq_\tR)$ denotes the far-field bistatic RCS of the scatterer $s$. 
Intuitively, the delay $\tau_s=(||\Delta\breve{\bq}_s||+||\Delta\tilde{\bq}_s||)/c$.
In the following sections, we derive an explicit expression of $\Xi_s(\bq_\tT,\bq_\tR)$ w.r.t. $\{\veps_s,\sigma_s\}$ by meshing the surface of scatterer $s$ into $M_s$ rectangular facets, and each is uniquely characterized by the parameter set:
\begin{equation}
    \cP_{m_s}=\left\{\bq_{m_s},\eta_{m_s},\vtheta_{m_s},\xi_{m_s},L_{\tH,m_s},L_{\tV,m_s},\veps_s,\sigma_s\right\},
    \label{equ:para_set}
\end{equation}
where $m_s\in\{1,2,\cdots,M_s\}$; $\bq_{m_s}\in\mathbb{R}^{3\times1}$ denotes the center position vector; $\eta_{m_s},\vtheta_{m_s},\xi_{m_s}$ are yaw, pitch, and roll angles, respectively; $L_{\tH,m_s},L_{\tV,m_s}$ describe the size of facet $m_s$. 
Thus, by defining $\Delta\breve{\bq}_s=\frac{1}{M_s}\sum_{m_s}\bq_{m_s}-\bq_\tT$, $\Xi(\bq_\tT,\bq_\tR)$ is given by \cite{ref:multi_facets_RCS}:
\begin{equation}
    \Xi_s(\bq_\tT,\bq_\tR)=\left|\sum_{m_s=1}^{M_s}\sqrt{\tS(\cP_{m_s},\ba_\tI,\ba_\tR)}e^{-j\frac{2\pi}{\lambda||\Delta\breve{\bq}_s||}\Delta\breve{\bq}_s^T\hat{\bn}_{m_s}}\right|^2,
    \label{equ:scatterer_RCS}
\end{equation}
where $\hat{\bn}_{m_s}\in\mathbb{R}^{3\times1}$ denotes the norm vector of facet $m_s$ pointing from interior to outside; $\ba_\tI, \ba_\tR\in\mathbb{R}^{3\times1}$ denote the unit incident and observation direction; $\tS(\cP_{m_s},\ba_\tI,\ba_\tR)$ denotes the bistatic RCS of facet $m_s$ with parameter set $\cP_{m_s}$.
\subsection{The closed-form expression of facet's RCS}\label{subsec:facets_geo_model}
As illustrated in Fig. \ref{fig:facets_geo_model}, we start from a simple scatterer with only 1 facet located at the origin of a 3D Cartesian coordinate system, with the parameter set $\cP=\{\mathbf{0}_{3\times1},0,0,0,L_\tH,L_\tV,\veps,\sigma\}$. 
$\mathbf{l}=[0,{l}_y,{l}_z]^T\in\mathbb{R}^{3\times1}$ denotes a point on the facet. 
Let incidence and observation directions be $\ba_\tI = {\Delta\breve{\bq}}/{||\Delta\breve{\bq}||} = [-\sin\theta_\tI\cos\phi_\tI, -\sin\theta_\tI\sin\phi_\tI, -\cos\theta_\tI]^{{T}}$ and $\ba_\tR={\Delta\tilde{\bq}}/{||\Delta\tilde{\bq}||}= [\cos\phi_\tR\sin\theta_\tR, \sin\phi_\tR\sin\theta_\tR, \cos\theta_\tR]^{{T}}$, where $\phi_\tI, \phi_\tR, \theta_\tI, \theta_\tR$ are effective azimuth and zenith angles of incidence and observation; $d_\tR=||\Delta\tilde{\bq}||$ denotes observation distance; $\gamma\in[0,\frac{\pi}{2}]$ denotes the polarization angle.
\begin{remark}\label{rem:incident_angle}
    \emph{
    In this paper, the reference directions of an azimuth angle $\phi$ and a zenith angle $\theta$ are defined by the positive $x$-axis and $z$-axis with counterclockwise direction as positive direction, thereby $\phi_\tI,\phi_\tR\in\left[-\frac{\pi}{2},\frac{\pi}{2}\right]$ and $\theta_\tI,\theta_\tR\in[0,\pi)$. 
    }
\end{remark}

The material properties are captured by the reflection coefficients given by:
\begin{lemma}\label{lem:ref_coef}
    The vertical and parallel polarization reflection coefficients are given by \cite{ref:EM_book}:
    \begin{equation}
        \Gamma_\perp(\veps,\sigma)=\frac{\cos\rho-\sqrt{\chi(\veps,\sigma)-\sin^2\rho}}{\cos\rho+\sqrt{\chi(\veps,\sigma)-\sin^2\rho}}\in\mathbb{C},
    \label{equ:ref_coeff_perp}
    \end{equation}
    \begin{equation}
        \Gamma_\parallel(\veps,\sigma)=\frac{\sqrt{\chi(\veps,\sigma)-\sin^2\rho}-\chi(\veps,\sigma)\cos\rho}{\sqrt{\chi(\veps,\sigma)-\sin^2\rho}+\chi(\veps,\sigma)\cos\rho}\in\mathbb{C},
        \label{equ:ref_coef_para}
    \end{equation}
    where $\chi(\veps,\sigma)=\veps-j\frac{\sigma}{\omega\veps_0}$ with $\veps_0\approx8.8542\times10^{-12}\;\text{F}/\text{m}$ and $\omega=2\pi f_c$; $\rho=\arccos(\cos\phi_\tI\sin\theta_\tI)$.
\end{lemma}
\begin{IEEEproof}
    Please refer to Appendix \ref{app:ref_coef}
\end{IEEEproof}

\begin{figure}[!t]
    \centering
    \includegraphics[width=0.5\linewidth]{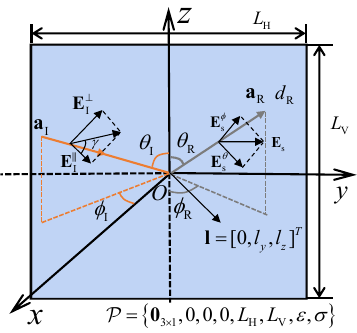}
    \caption{
    An illustration of a facets with size $L_\tH\times L_\tV$, whose relative permittivity and conductivity are denoted by $\veps$ and $\sigma$, respectively. 
    }
    \label{fig:facets_geo_model}
\end{figure}
The wireless signal is modeled as the plane EM wave, and the incident electric field phasor at $\bl$ is given by:
\begin{equation}
    \bE_\tI(\ba_\tI,\mathbf{l})=E_\tI e^{-j\frac{2\pi}{\lambda}\mathbf{a}_\tI^T\mathbf{l}}(\sin\gamma\hat{\bE}_\perp +\cos\gamma\hat{\bE}_\parallel),
    \label{equ:incident_E}
\end{equation}
where $E_\tI\in\mathbb{C}$;  $\hat{\bE}_\perp,\hat{\bE}_\parallel\in\mathbb{R}^{3\times1}$ are given by:
\begin{equation}
    \hat{\bE}_\perp=\frac{\hat{\bx}\times\ba_\tI}{||\hat{\bx}\times\ba_\tI||}=\frac{\cos\theta_\tI\hat{\by}-\sin\phi_\tI\sin\theta_\tI\hat{\bz}}{||\hat{\bx}\times\ba_\tI||},
    \label{equ:inci_perp_dir}
\end{equation}
\begin{equation}
    \hat{\bE}_\parallel=(\ba_\tI\times\hat{\bE}_\perp)/||\ba_\tI\times\hat{\bE}_\perp||.
    \label{equ:inci_para_dir}
\end{equation}
\begin{theorem}\label{the:current}
    After interacting with the facet, the induced surface current $\bJ(\bl)\in\mathbb{C}^{3\times1}$ at $\bl$ is given by:
    \begin{equation}
        \begin{aligned}
            \bJ(\bl)&= \left[\frac{\sin\gamma\cos\phi_\tI\sin\theta_\tI(1-\Gamma_\perp)}{||\ba_\tI\times\hat{\bE}_\perp||}(\cos\theta_\tI\hat{\by}-\sin\phi_\tI\sin\theta_\tI\hat{\bz})\right.\\
            &+\left.\frac{\cos\gamma(\Gamma_\para-1)}{||\ba_\tI\times\hat{\bx}||}(\cos\theta_\tI\hat{\bz}+\sin\phi_\tI\sin\theta_\tI\hat{\by})\right]\frac{E_\tI e^{-j\frac{2\pi}{\lambda}\ba_\tI^T\bl}}{Z_0},
        \end{aligned}
        \label{equ:current}
    \end{equation}
    where $\Gamma_\perp=\Gamma_\perp(\veps,\sigma),\Gamma_\parallel=\Gamma_\para(\veps,\sigma)$ are given in \emph{Lemma} \ref{lem:ref_coef}.
\end{theorem}
\begin{IEEEproof}
    Please refer to Appendix \ref{app:current}.
\end{IEEEproof}

Let $L_\tH$ and $L_\tV$ be chosen such that $||\Delta\breve{\bq}||$ and $||\Delta\tilde{\bq}||$ exceed the Rayleigh distance.
Thus, the scattered wave phasor is partitioned into two orthogonal components given by:
\begin{equation}
    \bE_\ts(\ba_\tR)\approx E_\ts^\phi(\ba_\tR)\hat{\bE}_\ts^\phi+E_\ts^\theta(\ba_\tR)\hat{\bE}_\ts^\theta,
    \label{equ:scattering_phasor}
\end{equation}
where $\bE_\ts(\ba_\tR)\in\mathbb{C}^{3\times1}$; 
$E_\ts^\phi(\ba_\tR),E_\ts^\theta(\ba_\tR)\in\mathbb{C}$ denote scattering wave phasors; $\hat{\bE}_\ts^\phi=-\sin\phi\hat{\bx}+\cos\phi\hat{\by}\in\mathbb{R}^{3\times1}$ and $\hat{\bE}_\ts^\theta=\cos\phi\cos\theta\hat{\bx}+\sin\phi\cos\theta\hat{\by}-\sin\theta\hat{\bz}\in\mathbb{R}^{3\times1}$ denote azimuth and polar directions. 
\begin{theorem}\label{the:scatter_phasor}
    Based on the far-field approximated Stratton-Chu equation \cite{ref:EM_book}, the scattering wave phasors $E_\ts^\phi(\ba_\tR),E_\ts^\theta(\ba_\tR)$ in \eqref{equ:scattering_phasor} are given by:
    \begin{equation}
        \begin{aligned}
            E_\ts^\phi(\ba_\tR)\simeq&-\frac{je^{-j\frac{2\pi}{\lambda}d_\tR}}{2d_\tR\lambda}E_\tI
            \left[
            \sin\gamma\Ups(\veps,\sigma)\Phi_\perp(\ba_\tI,\ba_\tR)\right.\\
            &\left.+\cos\gamma\Psi(\veps,\sigma)\Phi_\para(\ba_\tI,\ba_\tR)
            \right]I(\ba_\tI,\ba_\tR),
        \end{aligned}
        \label{equ:E_r_phi}
    \end{equation}
    \begin{equation}
        \begin{aligned}
            E_\ts^\theta(\ba_\tR)\simeq&-\frac{je^{-j\frac{2\pi}{\lambda}d_\tR}}{2d_\tR\lambda}E_\tI
            \left[
            \sin\gamma\Ups(\veps,\sigma)\Theta_\perp(\ba_\tI,\ba_\tR)\right.\\
            &\left.+\cos\gamma\Psi(\veps,\sigma)\Theta_\para(\ba_\tI,\ba_\tR)
            \right]I(\ba_\tI,\ba_\tR),
        \end{aligned}
        \label{equ:E_r_theta}
    \end{equation}
    where:
    \begin{equation}
        I(\ba_\tI,\ba_\tR)=L_\tH L_\tV\sinc\left(\frac{L_\tH}{\lambda}\zeta(\ba_\tI,\ba_\tR)\right)\sinc\left(\frac{L_\tV}{\lambda}\kappa(\ba_\tI,\ba_\tR)\right),
        \label{equ:inte}
    \end{equation}
    \begin{equation}
        \Phi_\perp(\ba_\tI,\ba_\tR)=\cos\phi_\tI\sin\theta_\tI\cos\theta_\tI\cos\phi_\tR/||\ba_\tI\times\hat{\bE}_\perp||,
        \label{equ:Phi_perp}
    \end{equation}
    \begin{equation}
        \begin{aligned}
            \Phi_\para(\ba_\tI,\ba_\tR)=\sin\phi_\tI\sin\theta_\tI\cos\phi_\tR/||\hat{\bx}\times\ba_\tI||,
        \end{aligned}
        \label{equ:Phi_para}
    \end{equation}
    \begin{equation}
        \begin{aligned}
            \Theta_\perp(\ba_\tI,\ba_\tR)=&\frac{\cos\phi_\tI\sin\theta_\tI}{||\ba_\tI\times\hat{\bE}_\perp||}(\cos\theta_\tI\sin\phi_\tR\cos\theta_\tR\\
            &+\sin\theta_\tI\sin\phi_\tI\sin\theta_\tR),
        \end{aligned}
        \label{equ:Theta_perp}
    \end{equation}
    \begin{equation}
    \begin{aligned}
        \Theta_\para(\ba_\tI,\ba_\tR)=\frac{-\cos\theta_\tI\sin\theta_\tR+\sin\phi_\tI\sin\theta_\tI\sin\phi_\tR\cos\theta_\tR}{||\hat{\bx}\times\ba_\tI||},
    \end{aligned}
        \label{equ:Theta_para}
    \end{equation}    
    with $\zeta(\ba_\tI,\ba_\tR)=\sin\phi_\tI\sin\theta_\tI+\sin\phi_\tR\sin\theta_\tR$; $\kappa(\ba_\tI,\ba_\tR)=\cos\theta_\tI+\cos\theta_\tR$; $\Ups(\veps,\sigma)=1-\Gamma_\perp(\veps,\sigma)$; $\Psi(\veps,\sigma)=\Gamma_\para(\veps,\sigma)-1$.
\end{theorem}
\begin{IEEEproof}
    Please refer to Appendix \ref{app:scatter_phasor}.
\end{IEEEproof}

Based on $E_\tI$ in \eqref{equ:incident_E} and $E_\ts$ in \eqref{equ:scattering_phasor}, the facet RCS in \eqref{equ:scatterer_RCS} is defined by:
\begin{small}
    \begin{equation}
        \tS(\cP,\ba_\tI,\ba_\tR)=\lim_{d_\tR\to\infty}4\pi d_\tR^2\frac{|E_\ts|^2}{|E_\tI|^2}\overset{(a)}{=}\lim_{d_\tR\to\infty}4\pi d_\tR^2\frac{|E_\ts^\phi|^2+|E_\ts^\theta|^2}{|E_\tI|^2},
        \label{equ:RCS_def}
    \end{equation}
\end{small}
where $\overset{(a)}{=}$ holds from the orthogonality in \eqref{equ:scattering_phasor}. 
\begin{corollary}\label{cor:RCS_facet}
The RCS of facet with $\cP$, incident direction $\ba_\tI$ and observation direction $\ba_\tR$ is given by:
    \begin{equation}
        \begin{aligned}
            \tS(\cP,\ba_\tI,\ba_\tR)=&\frac{\pi}{\lambda^2}\left[\sin^2\gamma(\Phi_\perp^2+\Theta_\perp^2)|\Ups(\veps,\sigma)|^2+\sin(2\gamma)\right.\\
            &\left.\times(\Phi_\perp\Phi_\para+\Theta_\perp\Theta_\para)\Re(\Ups(\veps,\sigma)\Psi^*(\veps,\sigma))\right.\\
            &\left.+\cos^2\gamma(\Phi_\para^2+\Theta_\para^2)|\Psi(\veps,\sigma)|^2\right]\left|I(\ba_\tI,\ba_\tR)\right|^2,
        \end{aligned}
        \label{equ:RCS_facet}
    \end{equation}
    where $\Phi_\perp,\Phi_\para,\Theta_\perp,\Theta_\para$ denote \eqref{equ:Phi_perp}-\eqref{equ:Theta_para}, respectively.
\end{corollary}
\begin{IEEEproof}
    Substituting equations in \emph{Theorem} \ref{the:scatter_phasor} into \eqref{equ:RCS_def} directly yields \eqref{equ:RCS_facet}.
\end{IEEEproof}
\begin{remark}\label{rem:mat_facet}
    \emph{
    Once the geometries of the target environment are determined, the RCSs in \eqref{equ:scatterer_RCS} and \eqref{equ:RCS_facet} only depend explicitly on the material properties $\bm$ through $\Ups(\veps,\sigma)$ and $\Psi(\veps,\sigma)$, which enables the inference of $\bm$ from $\cH$ via \eqref{equ:channel_s}.
    }
\end{remark}

\begin{figure}[!t]
    \centering
    \includegraphics[width=0.8\linewidth]{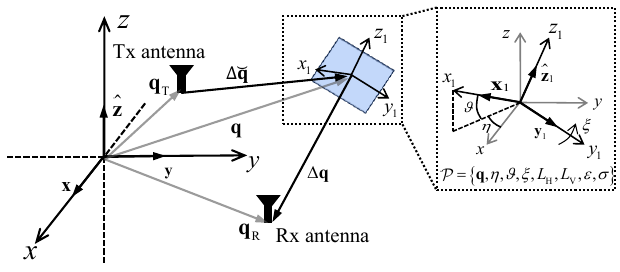}
    \caption{A facets with arbitrary position with $\cP$, where $||\bq||$ and $||\Delta\bq||$ satisfy the far-field constraints in \emph{Lemma} \ref{lem:rayleigh_dis}.}
    \label{fig:single_facets}
\end{figure}

Furthermore, as illustrated in Fig. \ref{fig:single_facets}, we consider the facet locates with arbitrary $\cP=\{\bq,\eta,\vtheta,\xi,L_\tH,L_\tV,\veps,\sigma\}$. 
The effective incident and observation vectors are given by:
\begin{equation}
    \ba_\tI(\cP,\bq_\tT)=-\mathbf{R}^T(\eta,\vtheta,\xi)\Delta\breve{\bq}/||\Delta\breve{\bq}||,
    \label{equ:inci_dir_arb}
\end{equation}
\begin{equation}
    \ba_\tR(\cP,\bq_\tT,\bq_\tR)=\mathbf{R}^T(\eta,\vtheta,\xi)\Delta\tilde{\bq}/||\Delta\tilde{\bq}||,
    \label{equ:obs_dir_arb}
\end{equation}
where $\bR(\eta,\vtheta,\xi)\in\mathbb{R}^{3\times3}$ denotes the unitary rotation matrix given in \eqref{equ:rotation_matrix} that satisfies $\bR\bR^T=\mathbf{I}_{3}$. 
Thus, the incident and observation angles are given by:
\begin{equation}
    \theta_\tI(\cP,\bq_\tT)=\arccos\left(-\ba_\tI^T(\cP)\hat{\bz}\right),
    \label{equ:incident_theta_local}
\end{equation}
\begin{equation}
    \theta_\tR(\cP,\bq_\tT,\bq_\tR)=\arccos(\ba_\tR^T(\cP,\bq_\tT,\bq_\tR)\hat{\bz}),
    \label{equ:receive_theta_local}
\end{equation}
\begin{equation}
    \phi_\tI(\cP,\bq_\tT)=\arcsin\left(-\ba_\tI^T(\cP)\hat{\by}/{\sin(\theta_\tI(\cP))}\right),
    \label{equ:incident_phi_local}
\end{equation}
\begin{equation}
    \phi_\tR(\cP,\bq_\tT,\bq_\tR)=\arcsin\left(\frac{\ba_\tR^T(\cP,\bq_\tT,\bq_\tR)\hat{\by}}{\sin(\theta_\tR(\cP,\bq_\tT,\bq_\tR))}\right).
    \label{equ:receive_phi_local}
\end{equation}
Besides, to ensure the facet locates at the far-field region defined by Rayleigh distance, the link distances $||\Delta\breve{\bq}||$ and $||\Delta\tilde{\bq}||$ should satisfy:
\begin{lemma}\label{lem:rayleigh_dis}
    To locates at the far-field region, $L_\tH,L_\tV$ of facet and $||\Delta\breve{\bq}||,||\Delta\tilde{\bq}||$ should be jointly considered:
    \begin{equation}
        \min\{||\Delta\breve{\bq}||,||\Delta\tilde{\bq}||\}\geq2(L_\tH^2+L_\tV^2)/\lambda-\lambda/32.
        \label{equ:rayleigh_dis}
    \end{equation}
\end{lemma}
\begin{IEEEproof}
    The Rayleigh distance $d_\Rayl$ is defined as the minimum link distance under normal incidence with $(\phi_\tI,\theta_\tI)=(0,\frac{\pi}{2})$, for which the phase error $\Delta\varphi(d_\Rayl)\leq\pi/8$:
    \begin{equation}
        \Delta\varphi(d_\Rayl)=\frac{2\pi}{\lambda}\left(\sqrt{d_\Rayl^2+\frac{L_\tH^2+L_\tV^2}{4}}-d_\Rayl\right)\leq\frac{\pi}{8}.
    \label{equ:app_phase_error}
    \end{equation}
    Therefore, we have $\min\{||\Delta\breve{\bq}||,||\Delta\tilde{\bq}||\}\geq\min\;d_\Rayl$ in \eqref{equ:rayleigh_dis}.
\end{IEEEproof}

\begin{figure*}
    \begin{equation}
            \bR(\eta,\vtheta,\xi)=
            \begin{bmatrix}
                \cos\eta & -\sin\eta & 0\\
                \sin\eta & \cos\eta & 0\\
                0 & 0 & 1
            \end{bmatrix}
            \times
            \begin{bmatrix}
                \cos\vtheta & 0 & \sin\vtheta\\
                0 & 1 & 0 \\
                -\sin\vtheta & 0 & \cos\vtheta
            \end{bmatrix}
            \times
            \begin{bmatrix}
                1 & 0 & 0 \\
                0 & \cos\xi & -\sin\xi \\
                0 & \sin\xi & \cos\xi
            \end{bmatrix}
            \label{equ:rotation_matrix}
    \end{equation}
    \rule{\textwidth}{.4pt}
\end{figure*}

The RCS of the facet in Fig. \ref{fig:single_facets} is obtained by substituting $\phi_\tI(\cP,\bq_\tT),\theta_\tI(\cP,\bq_\tT),\phi_\tR(\cP,\bq_\tT,\bq_\tR),\theta_\tR(\cP,\bq_\tT,\bq_\tR)$ from \eqref{equ:incident_theta_local}-\eqref{equ:receive_phi_local} into $\Ups(\veps,\sigma),\Psi(\veps,\sigma),\Phi_\perp,\Phi_\para,\Theta_\perp,\Theta_\para$ in \eqref{equ:Phi_perp}-\eqref{equ:Theta_para}, and then into \eqref{equ:RCS_facet}. 
Substituting $\tS(\cP_{m_s},\ba_\tI,\ba_\tR)$ for all $m_s\in\{1,\dots,M_s\}$ into \eqref{equ:scatterer_RCS} yields $\Xi(\bq_\tT,\bq_\tR)$ for all $s\in\{1,\dots,S\}$. 
Moreover, $\varphi_s\simeq\frac{2\pi}{\lambda}\left(||\Delta\breve{\bq}_s||+||\Delta\tilde{\bq}_s||\right)$ for practical simulation such as ray tracing \cite{ref:sionna}, since $||\Delta\breve{\bq}_s||,||\Delta\tilde{\bq}_s||\gg\lambda$ from \emph{Lemma} \ref{lem:rayleigh_dis}. 
In the following sections, an iterative algorithm is proposed to infer $\bm$ from $\cH$ and its performance is verified.

\section{Material Reconstruction}\label{sec:mat_recons}
In this section, an iterative gradient descent-based algorithm is proposed to solve the material reconstruction problem.

\subsection{Evaluation metric and loss function}
To comprehensively evaluate the difference between the channel $\hat{h}^{(t)}(\bq_\tT,\bq_n,\tau)$ induced by material $\bm^{(t)}\in\mathbb{R}^{2S\times1}$ based on \eqref{equ:overall_channel} and \eqref{equ:channel_s} at iteration $t\in\{1,2,\cdots,T_{\max}\}$ and the measurements $h(\bq_\tT,\bq_n,\tau)\in\cH$ via path gain and delay, the metric root mean square (RMS) delay spread is given by:
\begin{equation}
    \tau_\rms(h(\bq_\tT,\bq_\tR,\tau))=\sqrt{\frac{\sum_{s=1}^S|h_s(\bq_\tT,\bq_\tR)|^2(\tau_s-\bar{\tau})^2}{\sum_{s=1}^S|h_s(\bq_\tT,\bq_\tR)|^2}},
    \label{equ:RMS_delay}
\end{equation}
where $\bar{\tau}=\sum_{s=1}^S|h_s(\bq_\tT,\bq_\tR)|^2\tau_s/\sum_{s=1}^S|h_s(\bq_\tT,\bq_\tR)|^2$. 
The mean absolute square error (MASE) is given by:
\begin{equation}
    \cM(x,\hat{x})={|x-\hat{x}|^2}/{|x|^2}.
    \label{equ:MASE_def}
\end{equation}
Thus, the weighted loss function is given by:
\begin{equation}
    L(\cH,\bm^{(t)})=\frac{\sum_{n=1}^Nw_n\cM(\tau_\rms(\hat{h}^{(t)}),\tau_\rms(h))}{\sum_{n=1}^Nw_n},
    \label{equ:MASE}
\end{equation}
where $\hat{h}^{(t)}$ and $h$ denote $\hat{h}^{(t)}(\bq_\tT,\bq_\tR,\tau)$ and $h(\bq_\tT,\bq_\tR,\tau)$; the weight $w_n\in\mathbb{R}^+$ is given by:
\begin{equation}
    w_n=\exp(\cM(\tau_\rms(h(\bq_\tT,\bq_n,\tau),\tau_\rms(\hat{h}^{(t)}(\bq_\tT,\bq_n,\tau)))),
    \label{equ:weight}
\end{equation}
which aims to optimize the position with large deviation. 

\subsection{Optimization problem and the proposed algorithm}
The material reconstruction problem can be formulated as minimizing the MASE across $N$ samples by adjusting $\bm$:
\begin{equation}
    \begin{aligned}
        \min_{\bm}&\;L(\cH,\bm)\\
        \text{s.t.}&\;\veps_s\geq1,\;s\in\{1,2,\cdots,S\}\\
        &\;\sigma_s>0,\;s\in\{1,2,\cdots,S\}.
    \end{aligned}
    \label{equ:mat_recons_opt}
\end{equation}
Due to the complex expressions in \eqref{equ:channel_s}, \eqref{equ:scatterer_RCS}, \eqref{equ:RCS_facet}, \emph{Lemma} \ref{lem:ref_coef}, and the extra complicity from the square root and complex domain \cite{ref:convex_opt}, closed-form solutions are difficult to be obstained. 
We thus propose an iterative gradient descent algorithm to approximately solve \eqref{equ:mat_recons_opt} with tolerable complexity.
The gradient of loss function $L(\cH,\bm^{(t)})$ relative to material $\bm$ is given by:
\begin{equation}
    \nabla_\bm L(\cH,\bm^{(t)})=\left[\frac{\partial L}{\partial\veps_1},\frac{\partial L}{\partial\sigma_1},\cdots,\frac{\partial L}{\partial\veps_S},\frac{\partial L}{\partial\sigma_S}\right]^T\in\mathbb{R}^{2S\times1},
    \label{equ:grad}
\end{equation}
which can be easily obtained via auto gradient algorithms.
Thus, the gradient descent based update is formulated by:
\begin{equation}
    \bm^{(t+1)}=\mathcal{U}(\bm^{(t)}-\mu^{(t)}\nabla_\bm L(\cH,\bm^{(t)})),
    \label{equ:update}
\end{equation}
where $\mu^{(t)}=U\left(\frac{L(\cH,\bm^{(t)})}{L(\cH,\bm^{(0)})}\right)\mu^{(0)}\in[\mu_{\min},\mu_{\max}]$ denotes learning rate; $U(\cdot)$ and $\mathcal{U}(\cdot)$ are the step utility function that restricts $\mu^{(t)}$ and $\bm$.
Finally, based on previous sections, the material reconstruction algorithm is given in Algorithm \ref{algo:mat_recons}. 

\begin{algorithm}[htbp]
	\caption{Material reconstruction via gradient descent.}
	\label{algo:mat_recons}
	\begin{algorithmic}[1]
        \STATE {\textbf{Input:} $S$, $M_s$ for $\forall s$, $\cH$, $\mu^{(0)}$, $\mu_{\min}$, $\mu_{\max}$, $T_{\max}$, $\bq_\tT$;} 
		\STATE {Initialize $\bm^{(0)}\in\mathbb{R}^{2S\times1}$ and $t=0$;}
		\WHILE{$t<T_{\max}$}
		\STATE {Calculate $\hat{h}^{(t)}(\bq_\tT,\bq_n,\tau),\;n\in\{1,\cdots,N\}$ based on $\bm^{(t)}$ with \eqref{equ:channel_s}; }
		\STATE  {Calculate $L(\cH,\bm^{(t)})$ based on \eqref{equ:MASE};}
		\STATE {Calculate $\mu^{(t)}$ based on $L(\cH,\bm^{(t)})$ and $L(\cH,\bm^{(0)})$;}
        \STATE {Calculate $\nabla_\bm L(\cH,\bm^{(t)})$ in \eqref{equ:grad} via auto gradient;}
        \STATE {Update materials $\bm^{(t+1)}$ via \eqref{equ:update};}
        \STATE {Update iteration $t=t+1$;}
		\ENDWHILE
		\STATE {Obtain optimized materials $\bm^{T_{\max}}$;}
		\STATE {Obtain final loss function $L(\cH,\bm^{T_{\max}})$;}
		\STATE {\textbf{Output:} $\bm^{T_{\max}},\;L(\cH,\bm^{T_{\max}})$;}
	\end{algorithmic}
\end{algorithm}

\subsection{Quantitative metric for CKM generation}
CKM accuracy is quantitatively evaluated via the achievable rate $R(\bH)$ of an $M\times K$ MIMO system, where $\bH$ denotes the ground truth channel matrix and $\hat{\bH}_\ITU,\hat{\bH}_\tR$ are obtained by Sionna ray tracing \cite{ref:sionna} using default ITU materials and reconstructed materials, respectively.
Capacity of $\bH$ with noise power $N_0$ is given by:
\begin{equation}
    C(\bH)=\log_2\left(\left|\mathbf{I}_{M\times M}+\frac{1}{N_0}\bH\mathbf{V}\mathbf{Q}(\bH)\mathbf{V}^H\bH^H\right|\right),
    \label{equ:capacity}
\end{equation}
where $\bH=\mathbf{U}\mathbf{\Sigma}\mathbf{V}^H$ denotes singular value decomposition (SVD); $\mathbf{Q}\in\mathbb{R}^{K\times K}$ denote the diagonal water-filling power allocation matrix, where  $[\mathbf{Q}]_{(k,k)}=\max(\text{SNR}N_0/r+\sum_{n=1}^rN_0/rs_n^2-N_0/s_n^2,0)$, $r=\text{rank}(\mathbf{\Sigma})$, $s_n$ is the $n$th singular value of $\boldsymbol{\Sigma}$. 
Achievable rate with imperfect $\hat{\bH}$ generated by CKM at the transmitter side is given by \cite{ref:RIS_book}:
\begin{equation}
    R(\bH)=\log_2\left(\left|\mathbf{I}_{M\times M}+\frac{1}{N_0}\bH\hat{\mathbf{V}}\mathbf{Q}(\hat{\bH})\hat{\mathbf{V}}^H\bH^H\right|\right),
    \label{equ:achievable_rate}
\end{equation}
where $\hat{\bH}=\hat{\mathbf{U}}\hat{\mathbf{\Sigma}}\hat{\mathbf{V}}^H$.
The percentage error is given by:
\begin{equation}
    err(\bH,\hat{\bH})=(C(\bH)-R(\hat{\bH}))/C(\bH).
    \label{equ:error}
\end{equation}
Since $C(\bH)\geq R(\hat{\bH})$, we have $err(\bH,\hat{\bH})\in[0,1)$, and that of $N_\text{sam}$ samples are averaged to validate the accuracy of CKM.

\section{Simulation results}\label{sec:simulation}

The scatterer RCS in \eqref{equ:RCS_facet} is evaluated using the simulation parameters given by $f_c=3.5$ GHz, $L_\tH \times L_\tV = 0.39 \text{ m} \times 0.39 \text{ m}$, $\bq=[10 \text{ m},10 \text{ m},0]^T$, $\eta=\vtheta=\xi=0$, $\{\veps,\sigma\}=\{1.5,10\}$, $||\Delta\breve{\bq}||=d_\tR=10\sqrt{2} \text{ m}$, polarization $\gamma=\pi/2$, and $(\phi_\tI,\theta_\tI)=(-45^\circ,60^\circ)$, by comparing $\tS(\cP,\ba_\tI,\ba_\tR)$ with full wave physical optics (PO) and method of moments (MoM) results $\tS_\text{PO}(\cP,\ba_\tI,\ba_\tR)$ and $\tS_\text{MoM}(\cP,\ba_\tI,\ba_\tR)$ \cite{ref:EM_simu}. 
As shown in Fig. \ref{fig:RCS_compare}, the proposed model agrees well with full wave simulations, achieving a MASE of $-13.88$ dB relative to PO.


\begin{figure}[!t]
    \centering
    \includegraphics[width=0.4\linewidth]{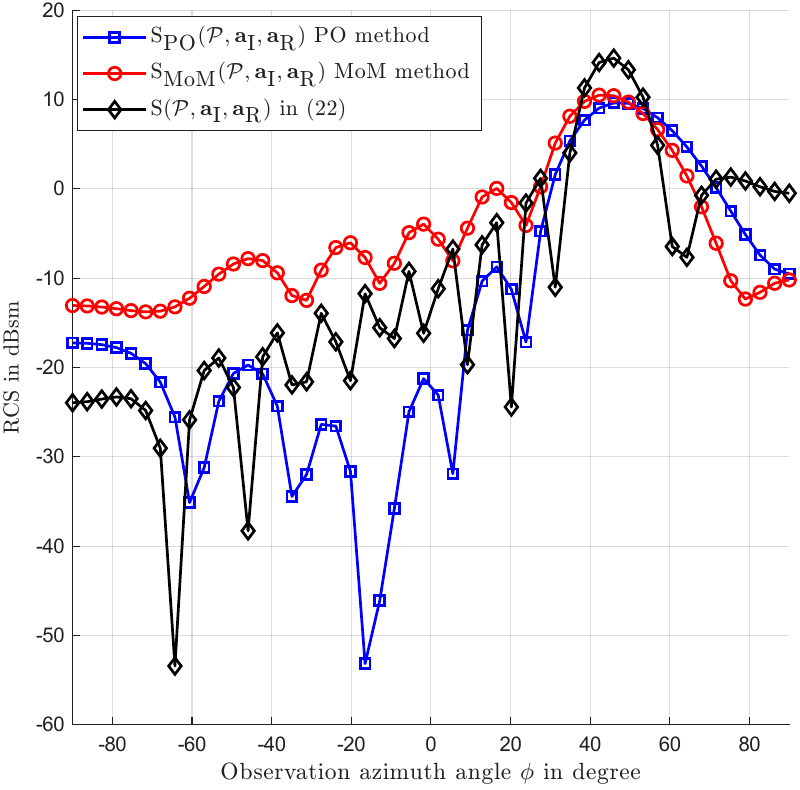}
    \caption{Comparison of the proposed RCS model with the full wave numerical results via PO and MoM with zenith observation angle $\theta_\tR=120^\circ$.}
    \label{fig:RCS_compare}
\end{figure}

The material reconstruction is evaluated in the 3D environment shown as the background in Fig. \ref{fig:loss}, where the the ground truth materials of $S=6$ scatterers follow ITU specifications \cite{ref:ITU_material}, with $f_c=3.5$ GHz, $N=5$, $\mu^{(0)}=0.04$, $T_{\max}=500$, $\bq_\tT=[0,0,0]^T$, $M\times K=8\times8$, $N_\text{sam}=1,000$, and SNR $=90$ dB. 
Comparing $err(\bH,\hat{\bH}_\ITU)$ for brick buildings and concrete ground with $err(\bH,\hat{\bH}_\tR)$ for reconstructed materials in Fig. \ref{fig:loss}, the average percentage error decreases from $11.66\%$ to $4.39\%$, and the median position-wise MASE of $\bH_\tR$ and $\bH$ is $-16.34$ dB, reveling effectiveness of the proposed algorithm. 



\begin{figure}[!t]
    \centering
    \includegraphics[width=\linewidth]{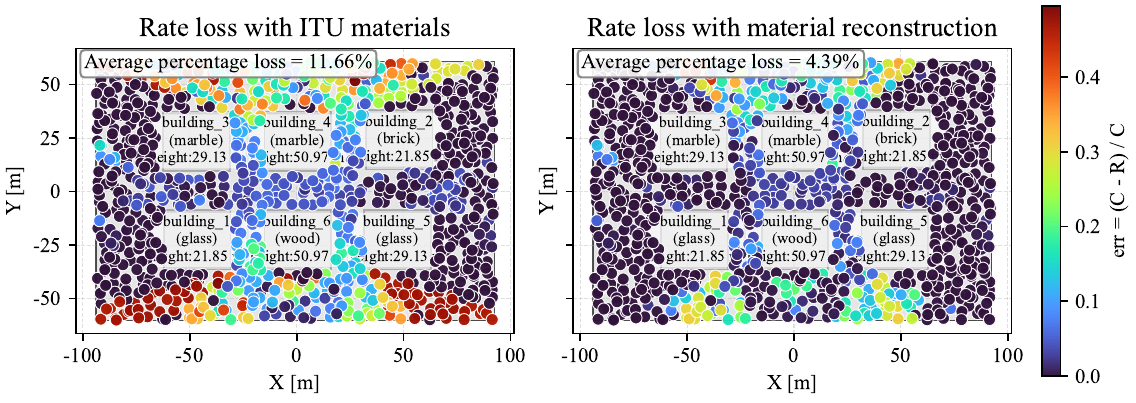}
    \caption{A Comparison between $err(\bH,\hat{\bH}_\ITU)$ and $err(\bH,\hat{\bH}_\tR)$.}
    \label{fig:loss}
\end{figure}

\section{Conclusion}\label{sec:con}
This paper presents a material-aware channel model that explicitly links material properties to the channel response, and proposes an iterative gradient descent material reconstruction algorithm validated by numerical simulations for efficient CKM generation.

{\appendices

\section{Proof of Lemma \ref{lem:ref_coef}}
\label{app:ref_coef}
Let $\rho$ denotes the angle between direction $\ba_\tI$ and $\hat{\bx}$, which is independent of polarization angle $\gamma$ and is given by:
\begin{equation}
    \rho=\arccos\left(|\hat{\bx}^T\ba_\ti(\phi_\ti,\theta_\ti)|\right)=\arccos(\cos\phi_\tI\sin\theta_\tI).
    \label{equ:app_incident_angle}
\end{equation}
Based on the Snell's law \cite{ref:EM_book}, we have:
\begin{equation}
    \Gamma_\perp(\veps,\sigma)=
    \frac{
    Z_0\cos\rho / \sqrt{\chi(\veps,\sigma)}-Z_0\cos\rho_\tet
    }{
    Z_0\cos\rho / \sqrt{\chi(\veps,\sigma)}+Z_0\cos\rho_\tet
    },
    \label{equ:app_ref_perp}
\end{equation}
\begin{equation}
    \Gamma_\para(\veps,\sigma)=
    \frac{
    -Z_0\cos\rho+Z_0\cos\rho_\tet/\sqrt{\chi(\veps,\sigma)}
    }{
    Z_0\cos\rho+Z_0\cos\rho_\tet/\sqrt{\chi(\veps,\sigma)}
    },
    \label{equ:app_ref_para}
\end{equation}
where $\sin\rho_\tet=\sqrt{1/\chi(\veps,\sigma)}\sin\rho$; $Z_0\approx377\;\Omega$ denotes the vacuum wave impedance. 
\eqref{equ:app_ref_perp} and \eqref{equ:app_ref_para} can direcly yield \eqref{equ:ref_coeff_perp} and \eqref{equ:ref_coef_para}.
Thus, the proof of \emph{Lemma} \ref{lem:ref_coef} is completed.

\section{Proof of Theorem \ref{the:current}}
\label{app:current}
The induced surface current at $\bl$ is given by:
\begin{equation}
    \begin{aligned}
        \bJ(\bl)=        \bJ_\perp(\bl)+\bJ_\para(\bl)\overset{(a)}{=}\hat{\bx}\times\left(\bH_\tI^\perp+\bH_\tr^\perp+\bH_\tI^\para+\bH_\tr^\para\right),
    \end{aligned}
    \label{equ:app_induced_current}
\end{equation}
where $\overset{(a)}{=}$ holds from PO \cite{ref:EM_book}; $\bH_\tI^\perp,\bH_\tr^\perp,\bH_\tI^\para,\bH_\tr^\para\in\mathbb{C}^{3\times1}$ are magnetic phasors for incident and reflected signal with vertical or parallel polarizations. 
After interacting, the reflected electric wave phasor at $\bl$  $\bE_\tr(\bl)\in\mathbb{C}^{3\times1}$ is given by Snell's law \cite{ref:EM_book}:
\begin{equation}
    \bE_\tr(\bl)=E_\tI e^{-j\frac{2\pi}{\lambda}\ba_\tI^T\bl}(\underbrace{\Gamma_\perp(\veps,\sigma)\sin\gamma\hat{\bE}_\tr^\perp}_{\bE_\tr^\perp\in\mathbb{C}^{3\times1}}+\underbrace{\Gamma_\para(\veps,\sigma)\cos\gamma\hat{\bE}_\tr^\para}_{\bE_\tr^\para\in\mathbb{C}^{3\times1}}),
    \label{equ:app_reflected_E}
\end{equation}
where $\hat{\bE}_\tr^\perp=\hat{\bE}_\perp\in\mathbb{R}^{3\times1}$.
Thus, we have \cite{ref:EM_book}:
\begin{equation}
    \bH_\tI^\perp=E_\tI\sin\gamma e^{-j\frac{2\pi}{\lambda}\ba_\tI^T\bl}(\ba_\tI\times\hat{\bE}_\perp)/Z_0||\ba_\tI\times\hat{\bE}_\perp||,
    \label{equ:app_H_I_perp}
\end{equation}
\begin{equation}
\bH_\tr^\perp=E_\tI\Gamma_\perp\sin\gamma e^{-j\frac{2\pi}{\lambda}\ba_\tI^T\bl} (\tilde{\ba}_\tI\times\hat{\bE}^\perp_\tr)/Z_0||\tilde{\ba}_\tI\times\hat{\bE}^\perp_\tr||,
    \label{equ:app_H_r_perp}
\end{equation}
\begin{equation}
    \bH_\tI^\para=E_\tI\cos\gamma e^{-j\frac{2\pi}{\lambda}\ba_\tI^T\bl}(\ba_\tI\times\hat{\bx})/Z_0||\ba_\tI\times\hat{\bx}||,
    \label{equ:app_H_I_para}
\end{equation}
\begin{equation}
    \bH_\tr^\para=-E_\tI\Gamma_\para\cos\gamma e^{-j\frac{2\pi}{\lambda}\ba_\tI^T\bl}(\ba_\tI\times\hat{\bx})/Z_0||\ba_\tI\times\hat{\bx}||,
    \label{equ:app_H_r_para}
\end{equation}
where $\tilde{\ba}_\tI=\ba_\tI\odot[-1,1,1]^T$. 
Substituting \eqref{equ:app_H_I_perp}-\eqref{equ:app_H_r_para} into \eqref{equ:app_induced_current} directly yields \eqref{equ:current}. 

\section{Proof of Theorem \ref{the:scatter_phasor}}
\label{app:scatter_phasor}
Based on the far-field approximated Stratton-Chu equation \cite{ref:EM_book}, the scattering wave phasors in \eqref{equ:scattering_phasor} are given by:
\begin{equation}
    E_\ts^\phi(\ba_\tI,\ba_\tR)\simeq -j\frac{kZ_0e^{-jkd_\tR}}{4\pi d_\tR}\underset{S}{\iint}\bJ^T(\bl)\hat{\bE}_\ts^\phi e^{jk\mathbf{a}_\tr^T\bl}\;\text{d}s,
    \label{equ:app_E_s_phi}
\end{equation}
and $E_\ts^\theta(\ba_\tI,\ba_\tR)$ is obtained by substituting $\hat{\bE}_\ts^\phi$ by $\hat{\bE}_\ts^\theta$ in \eqref{equ:app_E_s_phi}.
Substituting \eqref{equ:current} into \eqref{equ:app_E_s_phi} yields:
\begin{equation}
\begin{aligned}
    I(\ba_\tI,\ba_\tR)
    =\int_{-\frac{L_\tV}{2}}^{\frac{L_\tV}{2}}\int_{-\frac{L_\tH}{2}}^{\frac{L_\tH}{2}}e^{j\frac{2\pi}{\lambda}(\mathbf{a}_\tR-\mathbf{a_\tI})^T\bl}\;\text{d}l_y\text{d}l_z,
\end{aligned}
    \label{equ:app_inte}
\end{equation}
which yields \eqref{equ:inte}.
Thus, the proof of \emph{Theorem} \ref{the:scatter_phasor} is completed. 


}

\bibliography{myrefs}

\end{document}